\documentclass[12pt,letter]{article}

\usepackage{amsfonts,amssymb,amsmath,amsthm,bm}
\usepackage{latexsym,url,color,enumerate}
\usepackage{graphicx}
\usepackage{kbordermatrix}

\newtheorem{theorem}{Theorem}
\newtheorem{proposition}{Proposition}
\newtheorem{lemma}{Lemma}
\newtheorem{corollary}{Corollary}

\newtheorem{example}{Example}
\newtheorem{remark}{Remark}

\newcommand{\footremember}[2]{%
    \footnote{#2}
    \newcounter{#1}
    \setcounter{#1}{\value{footnote}}%
}
\newcommand{\footrecall}[1]{%
    \footnotemark[\value{#1}]%
} 
\newcommand\nc\newcommand
\nc\bfa{{\boldsymbol a}}\nc\bfA{{\bf A}}\nc\cA{{\mathcal A}}
\nc\bfb{{\boldsymbol b}}\nc\bfB{{\bf B}}\nc\cB{{\mathcal B}}
\nc\bfc{{\boldsymbol c}}\nc\bfC{{\bf C}}\nc\cC{{\mathcal C}}
\nc\bfd{{\boldsymbol d}}\nc\bfD{{\bf D}}\nc\cD{{\mathcal D}}
\nc\bfe{{\boldsymbol e}}\nc\bfE{{\bf E}}\nc\cE{{\mathcal E}}
\nc\bff{{\boldsymbol f}}\nc\bfF{{\bf F}}\nc\cF{{\mathcal F}}
\nc\bfg{{\boldsymbol g}}\nc\bfG{{\bf G}}\nc\cG{{\mathcal G}}
\nc\bfh{{\boldsymbol h}}\nc\bfH{{\bf H}}\nc\cH{{\mathcal H}}
\nc\bfi{{\boldsymbol i}}\nc\bfI{{\bf I}}\nc\cI{{\mathcal I}}
\nc\bfj{{\boldsymbol j}}\nc\bfJ{{\bf J}}\nc\cJ{{\mathcal J}}
\nc\bfk{{\boldsymbol k}}\nc\bfK{{\bf K}}\nc\cK{{\mathcal K}}
\nc\bfl{{\boldsymbol l}}\nc\bfL{{\bf L}}\nc\cL{{\mathcal L}}
\nc\bfm{{\boldsymbol m}}\nc\bfM{{\bf M}}\nc\cM{{\mathcal M}}
\nc\bfn{{\boldsymbol n}}\nc\bfN{{\bf N}}\nc\cN{{\mathcal N}}
\nc\bfo{{\boldsymbol o}}\nc\bfO{{\bf O}}\nc\cO{{\mathcal O}}
\nc\bfp{{\boldsymbol p}}\nc\bfP{{\bf P}}\nc\cP{{\mathcal P}}
\nc\bfq{{\boldsymbol q}}\nc\bfQ{{\bf Q}}\nc\cQ{{\mathcal Q}}
\nc\bfr{{\boldsymbol r}}\nc\bfR{{\bf R}}\nc\cR{{\mathcal R}}
\nc\bfs{{\boldsymbol s}}\nc\bfS{{\bf S}}\nc\cS{{\mathcal S}}
\nc\bft{{\boldsymbol t}}\nc\bfT{{\bf T}}\nc\cT{{\mathcal T}}
\nc\bfu{{\boldsymbol u}}\nc\bfU{{\bf U}}\nc\cU{{\mathcal U}}
\nc\bfv{{\boldsymbol v}}\nc\bfV{{\bf V}}\nc\cV{{\mathcal V}}
\nc\bfw{{\boldsymbol w}}\nc\bfW{{\bf W}}\nc\cW{{\mathcal W}}\nc\sW{{\mathscr W}}
\nc\bfx{{\boldsymbol x}}\nc\bfX{{\boldsymbol X}}\nc\cX{{\mathcal X}}\nc\sX{{\mathscr X}}
\nc\bfy{{\boldsymbol y}}\nc\bfY{{\boldsymbol Y}}\nc\cY{{\mathcal Y}}\nc\sY{{\mathscr Y}}
\nc\bfz{{\boldsymbol z}}\nc\bfZ{{\boldsymbol Z}}\nc\cZ{{\mathcal Z}}\nc\sZ{{\mathscr Z}}
\newcommand{\RR}{\mathbb{R}}%Real numbers
\newcommand{\FF}{\mathbb{F}}%Finite field
\newcommand{\rank}{\textup{rank}}
\newcommand{\minrk}{\textup{minrk}}

\newcommand{\chrnb}{\overline{\chi}}
\begin{document}
\title{A Bound on the Shannon Capacity via a Linear Programming Variation\thanks{This paper was presented in part at 2017 IEEE International Symposium on Information Theory.} }
\author{
Sihuang Hu\footremember{alley}{Lehrstuhl D f\"ur Mathematik, RWTH Aachen, Germany (husihuang@gmail.com). 
This work was done while S. Hu was with Department of Electrical Engineering - Systems, Tel Aviv University, Israel. 
Research supported by ERC grant no.~639573, ISF grant no.~1367/14, and the Alexander von Humboldt Foundation.}
\and Itzhak Tamo\footremember{trailer}{Department of Electrical Engineering--Systems, 
  Tel Aviv University, Tel Aviv, Israel (tamo@post.tau.ac.il, ofersha@eng.tau.ac.il).
  The work of I. Tamo was supported by ISF grant no.~1030/15 and NSF-BSF grant no.~2015814.
  The work of O. Shayevitz was supported by ERC grant no.~639573, and ISF grant no.~1367/14.
  }
\and Ofer Shayevitz\footrecall{trailer} 
}

\date{}
\maketitle

\begin{abstract}
We prove an upper bound on the Shannon capacity of a graph via a linear programming variation. We show that our bound can outperform both the Lov\'asz theta number and the Haemers minimum rank bound. As a by-product, we also obtain a new upper bound on the broadcast rate of Index Coding.
\end{abstract}

\section{Introduction}
Let $G=(V(G),E(G))$ be an undirected graph. An \textit{independent set} in $G$ is a subset of pairwise non-adjacent vertices. The \emph{independence number} of $G$, denoted by $\alpha(G)$, is the largest possible size of an independent set in $G$. For two graphs $G$ and $H$, their \emph{strong product} $G\boxtimes H$ is a graph such that
\begin{enumerate}
  \item the vertex set of  $G\boxtimes H$ is the Cartesian product $V(G) \times V(H)$; and
  \item any two distinct vertices $(u,u')$ and $(v,v')$ are adjacent in $G\boxtimes H$ if
    $u\sim v$ and $u'=v'$, or $u=v$ and $u'\sim v'$, or $u\sim v$ and $u'\sim v'$.
\end{enumerate}
The graph $G^n$ is defined inductively by $G^{n}= G^{n-1} \boxtimes G$.
The \emph{Shannon capacity} of a graph $G$ is defined by
\begin{align}\label{eq:shannon_cap_def}
  \Theta(G)
  :=\sup_{n}\sqrt[\leftroot{-3}\uproot{3}n]{\alpha(G^n)}
  =\lim_{n\to\infty}\sqrt[\leftroot{-3}\uproot{3}n]{\alpha(G^n)}
\end{align}
where the limit exists by the supermultiplicativity of $\alpha(G^n)$ and Fekete's Lemma. 

This graph quantity was introduced by Shannon~\cite{Sh1956} as the zero-error capacity in the context of channel coding. 
In this setup, a transmitter would like to communicate a message to a receiver through the channel, and the receiver must decode the message without error. 
This problem can be equivalently cast in terms of the
\textit{confusion graph} $G$ associated with the channel. The vertices of the confusion graph are the input symbols, and two vertices are adjacent if the corresponding inputs can result in the same output. 
It is easy to check that $G^n$ is the confusion graph for $n$ uses of the channel, and that $\alpha(G^n)$ is the maximum number of messages that can be transmitted without error over $n$ uses of the channel.

Despite the apparent simplicity of the problem, a general characterization of $\Theta(G)$ remains elusive. Several lower and upper bounds were obtained by Shannon~\cite{Sh1956}, Lov\'asz~\cite{Lo1979} and Haemers~\cite{Ha1979}. These bounds are briefly reviewed in Section~\ref{sec:known bounds}. 
In Section~\ref{sec:new bound} we present a new bound on the Shannon capacity via a variation on the linear program pertaining to the fractional independence number of the graph.
Next, we show that the new bound can simultaneously outperform both the Lov\'asz theta number and the Haemers minimum rank bound. 
In Section~\ref{sec:IndexCoding}, we leverage our bound to prove a new upper bound for the broadcast rate of Index Coding. 
It should be noted that a fractional version of the Haemers minimum rank bound, denoted $\minrk_f^{\FF}$, was introduced independently 
by Blasiak~\cite{Blasiak2013} and Shanmugam et al.~\cite{ShanmugamAsterisDimakis2015}, and investigated in more detail by Bukh and Cox~\cite{BukhCox2018} very recently. 
This bound is at least as good as $\minrk_{\FF}^*$, one of our new bounds. Nevertheless $\minrk_f^{\FF}$ is very difficult to compute,
and our $\minrk_{\FF}^*$ bound is more tractable and provides a feasible way to approach $\minrk_f^{\FF}$ (see Remark~\ref{FractionalHaemersBound} below for more details).

\section{Upper Bounds on the Shannon Capacity}\label{sec:known bounds}
In this section, we give a brief overview of three well-known upper bounds on the Shannon capacity. Throughout this section let $G$ be a graph with vertex set $V(G)=\{1,2,\dots,m\}$. 
\subsection{Fractional Independence Number}
The fractional independence number is the linear programming relaxation of the $0$-$1$ integer linear programming that computes the independence number.
More precisely, the \emph{fractional independence number} $\alpha_f(G)$ is defined as the optimal value of the following linear program:
\begin{align}
  \label{Primal}
  \begin{split}
 \textup{maximize } &\sum_{x} w(x) \\ 
 \textup{subject to }  & \sum_{x\in C}w(x) \le 1 \textup{ for every clique $C$ in $G$},\\ 
                & w(x)\ge 0.
  \end{split}
\end{align}
(A clique $C$ in $G$ is a subset of the vertices, $C\subseteq V(G)$, such that every two distinct vertices are adjacent in $G$.)
From the duality theorem of linear programming, $\alpha_f(G)$ can also be computed as follows:
\begin{align}
  \label{Dual}
  \begin{split}
 \textup{minimize } &\sum_{C} q(C) \\ 
 \textup{subject to }  &\sum_{C\owns x}q(C) \ge 1 \textup{ for each vertex $x$ in $G$},\\
                &q(C)\ge 0.
              \end{split}
\end{align}
(The optimal value of~\eqref{Dual} is also called the \emph{fractional clique-cover number} of $G$, and denoted as $\overline{\chi}_f(G)$.)

The following bound was first given by Shannon~\cite{Sh1956}, and was also discussed in detail by Rosenfeld~\cite{Ro1967}.
\begin{theorem}\cite[Theorem 7]{Sh1956}
  \label{ShannonBound}
  $\Theta(G)\le\alpha_f(G)$.
\end{theorem}

\subsection{Lov\'asz Theta Number}
In his seminal paper~\cite{Lo1979}, Lov\'asz solved the long-standing problem of the Shannon capacity of the pentagon graph, by introducing an important new graph invariant, called the Lov\'asz theta number.
An \emph{orthonormal representation of $G$} is a system of unit vectors $\bfv_1,\dots,\bfv_m$ in some Euclidean space $\RR^d\ (d\ge 1)$ such that if $i$ and $j$ are nonadjacent then $\bfv_i$ and $\bfv_j$ are orthogonal (all vectors will be column vectors). The \emph{Lov\'asz theta number} of $G$ is defined as
\begin{align*}
  \vartheta(G):=\min_{\bfc,\bfv_i}\max_{1\le i\le m}\frac{1}{(\bfc^T\bfv_i)^2}
\end{align*}
where the minimum is taken over all unit vectors $\bfc$ and all orthonormal representations $\{\bfv_1,\dots,\bfv_m\}$ of $G$. 
The following bound is the main result of~\cite{Lo1979}.
\begin{theorem}\cite[Theorem 1]{Lo1979} 
  \label{LovaszThetaNumber}
$\Theta(G)\le\vartheta(G).$
\end{theorem}

In the sequel we will also need the following results from~\cite{Lo1979}. The theta number of odd cycles was calculated by Lov\'asz~\cite{Lo1979}.
\begin{proposition}\cite[Corollary 5]{Lo1979}
  \label{LovaszThetaNumberOfOddCycles}
  For odd $n$,
  \begin{align*}
    \vartheta(C_n) = \frac{n\cos(\pi/n)}{1+\cos(\pi/n)}.
  \end{align*}
\end{proposition}

In particular, for the pentagon graph, Lov\'asz proved that $\Theta(C_5)\le\vartheta(C_{5})=\sqrt{5}$, which meets the lower bound given by Shannon~\cite{Sh1956}.
Also, there exists the following duality between $G$ and its complementary graph $\overline{G}$.

\begin{proposition}\cite[Theorem 5]{Lo1979}
  \label{LovaszThetaDuality}
  Let $\bfd$ range over all unit vectors and let $\{\bfu_1,\dots,\bfu_m\}$ range over all orthonormal representations of $\overline{G}$. Then
  \begin{align}\label{LovaszDual}
    \vartheta(G)=\max_{\bfd,\bfu_i}\sum_{i=1}^{m}(\bfd^T\bfu_i)^2.
  \end{align}
\end{proposition}

For two graphs $G=(V(G),E(G))$ and $H=(V(H),E(H))$, their \emph{disjoint union}, denoted as $G+H$, is the graph whose vertex set is the disjoint union of $V(G)$ and $V(H)$ and whose edge set is the disjoint union of $E(G)$ and $E(H)$.  
The Lov\'asz theta number is multiplicative with respect to the strong product,
and additive with respect to the disjoint union.
\begin{proposition}\label{prop:LovaszProperty}
  \ \\[-5mm]
  \begin{enumerate}
    \item 
  \cite[Theorem 7]{Lo1979}
  $\vartheta(G\boxtimes H)=\vartheta(G)\cdot\vartheta(H)$.
\item
  \cite[Section 18]{Knuth1994}
  $\vartheta(G+H)=\vartheta(G)+\vartheta(H)$.
  \end{enumerate}
\end{proposition}

\subsection{Haemers Minimum Rank Bound}\label{Minrank}
Haemers~\cite{Ha1979,Ha1981} proved a very useful upper bound based on the matrix rank as follows.
An $m\times m$ matrix $B$ over some field is said to \emph{fit} $G$ if $B_{ii}\ne0$ for $1\le i\le m$, and $B_{ij}=0$ whenever vertices $i$ and $j$ are nonadjacent for $1\le i,j\le m$ and $i\ne j$. Let $B^{\otimes n}$ denote the Kronecker product of $n$ copies of $B$. It is easy to verify that if $B$ fits $G$, then $B^{\otimes n}$ fits $G^n$.

\begin{theorem}\cite{Ha1981}
  \label{HaemersBound}
  If a matrix $B$ fits a graph $G$, then $\Theta(G)\le\rank(B)$.
\end{theorem}

For a graph $G$, Haemers~\cite{Ha1981} introduced the following graph invariant 
\begin{align*}
  \minrk(G):=\min\{\rank(B): B\textup{ fits }G\},
\end{align*}
where the minimization is taken over all fields.
By Theorem~\ref{HaemersBound} it follows that $\Theta(G)\le\minrk(G)$. Moreover, for a fixed field $\FF$, define
\begin{align*}
  \minrk_{\FF}(G):=\min\{\rank(B): B\textup{ over $\FF$ fits }G\}.
\end{align*}
It is easy to verify that $\minrk_{\FF}$ is submultiplicative with respect to the strong product and additive with respect to the disjoint union, i.e., for any two graphs $G$ and $H$,
\begin{align*}
\minrk_{\FF}(G\boxtimes H)&\le\minrk_{\FF}(G)\cdot\minrk_{\FF}(H),\\
\minrk_{\FF}(G + H)&=\minrk_{\FF}(G)+\minrk_{\FF}(H).
\end{align*}
The following example is provided by Haemers~\cite{Ha1979} to answer some problems raised in~\cite{Lo1979}.

\begin{example}{\rm\cite{Ha1979}\label{ex:Schlaefli}
    Let $G$ be the complement of the Schl\"afli graph, which is the unique strongly regular 
    graph\footnote{A strongly regular graph with parameters $(v,k,\lambda,\mu)$ is a regular graph with $v$ vertices and degree $k$
    such that every two adjacent vertices have $\lambda$ common neighbours, and
every two non-adjacent vertices have $\mu$ common neighbours.} 
    with parameters $(27,16,10,8)$. Let $A$ be the adjacency matrix of $G$, and let $I$ be the identity matrix of order $27$. 
    Then the matrix $A-I$ fits the graph $G$, and its rank over $\RR$ is equal to $7$. Hence $\minrk(G)\le\minrk_{\RR}(G)\le7<9=\vartheta(G)$. This improves the bound given by the Lov\'asz theta number. Moreover, Tims~\cite[Example 3.8]{Tims2013} proved that $\minrk_{\FF}(G)\ge7$ over any field $\FF$, and therefore $\minrk(G)=7$. Similarly, the rank of the matrix $A-I$ over the field $\FF_{11}$ is also equal to $7$, hence $\minrk_{\FF_{11}}(G)=7$ (this fact will be used in Example~\ref{IndexCodingSchlaefliModify} and Remark~\ref{BukhCox}).
}
\end{example}

\section{A Linear Programming Variation}\label{sec:new bound}
In this section we will prove our main result, providing a new upper bound on the Shannon capacity by a variation of the linear programming bound given in~\eqref{Primal}-\eqref{Dual}.
For a subset $S\subset V(G)$, the \emph{induced subgraph} $G_S$ is the graph whose vertex set is $S$ and whose edge set consists of all of the edges in $E$ that have both endpoints in $S$.
%Let $f$ be an upper bound on the independence number, that is, $\alpha(G)\le f(G)$ for any graph $G$.
%Now we state our variation of the linear programs~\eqref{Primal}-\eqref{Dual} using such an upper bound $f$. 

Let $f$ be a real-valued function defined on graphs, and 
let $f^*(G)$ be the optimal value of the following linear program:
\begin{align}\label{PrimalEx}
  \begin{split}
 \textup{maximize } &\sum_{x} w(x) \\ 
 \textup{subject to }  &\sum_{x\in S}w(x)\le f(G_S)\textup{ for each subset $S$ of } V(G),\\ 
  & w(x)\ge 0.
 \end{split}
\end{align}
By duality $f^*(G)$ can also be computed as follows:
\begin{align}\label{DualEx}
  \begin{split}
  \textup{minimize } &\sum_{S} q(S)f(G_S) \\ 
 \textup{subject to }  &\sum_{S\owns x}q(S) \ge 1 \textup{ for each vertex $x$ in $G$},\\
                &q(S)\ge 0.
  \end{split}
\end{align}

\begin{remark}\label{remark:1}
  {\rm
    The non-negative real-valued function $q:2^{V(G)}\rightarrow \RR$ in~\eqref{DualEx}, satisfying that $\sum_{S\owns x}q(S) \ge 1$ for each vertex $x$ in $G$, is called a \emph{fractional cover} of $G$ by 
    %Fachini and Korner~\cite{FachiniKorner2000} and 
K\"orner, Pilotto and Simonyi~\cite{KornerPilottoSimonyi2005}. The fractional cover is used to generalize the local chromatic number to provide an upper bound for the Sperner capacity of directed graphs\footnote{The Sperner capacity of directed graphs is a natural generalization of the Shannon capacity of undirected graphs. See~\cite{KornerPilottoSimonyi2005} for definitions of the local chromatic number and the Sperner capacity.}, cf.~\cite[Theorem
6]{KornerPilottoSimonyi2005}. Note that for undirected graphs, the bounds in~\cite{KornerPilottoSimonyi2005} are always no stronger than the fractional independence number, and hence are not useful upper bounds for the Shannon capacity.
}
\end{remark}

By taking $q(S)=1$ for $S=V(G)$ and $q(S)=0$ otherwise, it is readily verified that $f^*(G)\leq f(G)$. In the next two lemmas, we show that if the function $f$ satisfies certain properties, then these properties are also inherited by $f^*$. We say that $f$ is an upper bound on the independence number if $\alpha(G)\leq f(G)$ for any graph $G$. 
%More precisely, we prove that  if $f$ is an upper bound on the independence number, that is, $\alpha(G)\leq f(G)$ for any graph $G$, then also $f^*$ is an upper bound on the independence number. Similarly, we prove it for the submultiplicative property that will be defined later.  

\begin{lemma}\label{UpperBounded}
  If $f$ is an upper bound on the independence number, then so is $f^*$.
\end{lemma}
\begin{proof}
  This result can be proved directly using the primal linear program~\eqref{PrimalEx}. However we would like to present a different proof using the dual linear program~\eqref{DualEx} and a counting argument as follows.
  Let $\Gamma$ be any independent set in $G$ and $q$ a fractional cover of $G$. Since $f(G_S)\ge\alpha(G_S)$, we have
\begin{align*}
  \sum_{S}q(S)f(G_S) 
  &\ge \sum_{S}q(S)\alpha(G_S) \\
  &\ge \sum_{S}q(S)\alpha(G_{S\cap\Gamma}) \\
  &= \sum_{x\in\Gamma}\sum_{S\owns x}q(S)\\
  &\ge |\Gamma|.
\end{align*}
This proves the result.
\end{proof}

We say that $f$ is \emph{submultiplicative} (with respect to the strong product) if for any two graphs $G$ and $H$, $f(G\boxtimes H)\le f(G)f(H)$.

\begin{lemma}
  \label{Submultiplicative}
  If $f$ is submultiplicative, then so is $f^*$.
\end{lemma}
\begin{proof}
  Let $q_1$ and $q_2$ be optimal solutions of the linear program~\eqref{DualEx} for $G$ and $H$ respectively. Now we assign weights $q(W)$ to each subset $W$ of $V(G\boxtimes H)$ as follows: if $W=S\times T$ for some $S\subset V(G)$ and $T\subset V(H)$, then we set $q(W)=q(S\times T)=q_1(S)q_2(T)$; otherwise we set $q(W)=0$.  Then for each vertex $(x,y)$ in $V(G\boxtimes H)$, we have
  \begin{align*}
    \sum_{W\owns (x,y)} q(W) &= \sum_{S\owns x, T\owns y}q(S\times T)\\
    &=\sum_{S\owns x}q_1(S)\sum_{T\owns y}q_2(T)\\
    &\ge1.
  \end{align*}
So $q$ is a feasible solution for~\eqref{DualEx}, and
\begin{align*}
  f^*(G\boxtimes H) \le&\sum_{W}q(W)f((G\boxtimes H)_{W})\\
  =&\sum_{S,T}q(S\times T)f((G\boxtimes H)_{S\times T})\\
  =&\sum_{S,T}q_{1}(S)q_{2}(T)f(G_S\boxtimes H_T)\\
  \le& \sum_{S}q_1(S)f(G_S)\sum_{T}q_2(T)f(H_{T})\\
  =&f^*(G) \cdot f^*(H)
\end{align*}
in which the second inequality follows from the submultiplicativity of $f$.
This proves the result.
\end{proof}

Now we can prove the following upper bound on the Shannon capacity.
\begin{theorem}\label{NewBound}
Let $f$ be a submultiplicative upper bound on the independence number. Then, $$\Theta(G)\le f^*(G).$$
\end{theorem}
\begin{proof}
  By Lemma~\ref{UpperBounded} and Lemma~\ref{Submultiplicative}, we get $\alpha(G^n)\le f^*(G^n)\le f^*(G)^n$.
\end{proof}

Any function $f$ that is a submultiplicative upper bound on the independence number forms by itself an upper bound on the Shannon capacity, i.e., $\Theta(G)\leq f(G)$. Combining this with Theorem~\ref{NewBound} and the fact that $f^*(G)\leq f(G)$ we get $\Theta(G) \leq f^*(G) \leq f(G).$ 
Simply put, this chain of inequalities shows that $f^*$
is a bound that is at least as good as the bound $f$ that we started with in the first place. An immediate question is, can we get the strict inequality $f^*(G)<f(G)$? In other words, can we improve the bound $f$ on the Shannon capacity by solving the corresponding linear programming problem? In the sequel, we give an affirmative answer to this question by providing several explicit examples where a strict inequality holds. 
Furthermore, we answer the following two natural questions: 1) which functions $f$ should we use in Theorem~\ref{NewBound}? and 2) do we always get a tighter upper bound for any function $f$?   

Before we proceed to answer those questions, we show some simple properties of $f^*$, which are used later. 
We say that $f$ is \emph{superadditive} with respect to the disjoint union if $f(G+H)\ge f(G)+f(H)$ for any two graphs $G$ and $H$.
\begin{proposition}\label{prop:properties}
  \ \\[-5mm]
  \begin{enumerate}
    \item If $f(C)=1$ for each clique $C$ in $G$, then $f^*(G)\le\alpha_f(G)$. In particular, $\minrk_{\FF}^*(G)\le\alpha_f(G)$.
    \item $f^*(G+H)\le f^*(G)+f^*(H)$.
    \item If $f$ is superadditive, then $f^*(G+H)=f^*(G)+f^*(H)$. In particular, $\minrk_{\FF}^*(G+H)=\minrk_{\FF}^*(G)+\minrk_{\FF}^*(H).$
  \end{enumerate}
\end{proposition}
\begin{proof}
  1) Follows directly from~\eqref{PrimalEx}. 
  2) Follows directly from~\eqref{DualEx}. 
  3) Let $w_1$ and $w_2$ be optimal solutions of the primal linear program~\eqref{PrimalEx} for $G$ and $H$ respectively. We define an assignment $w$ for $G+H$ as follows: $w(x)=w_1(x)$ if $x\in V(G)$ and $w(y)=w_2(y)$ if $y\in V(H)$. By the superadditivity of $f$, we can verify that $w$ is a feasible solution of~\eqref{PrimalEx} for $G+H$, and thus $f^*(G+H)\ge f^*(G)+f^*(H)$. Combining it with 2) proves that $f^*(G+H)=f^*(G)+f^*(H)$. 
  The second equality follows from the fact that $\minrk_{\FF}$ is additive with respect to the disjoint union. 
\end{proof}

\subsection{A New Bound $\minrk_{\FF}^*$}
Now we take $f=\minrk_{\FF}$ which is a submultiplicative upper bound on the Shannon capacity, and show that there exist graphs such that our new bound $\minrk_{\FF}^*$ can outperform both $\minrk$ and Lov\'asz theta number. 
The following three examples show several instances of it.
Example~\ref{ex:C_n} shows a family of graphs where our bound outperforms $\minrk$ but not Lov\'asz theta number.

\begin{example}\label{ex:C_n}
  {\rm
    For odd $n\ge5$, it is not hard to verify that $\minrk(C_n)=\minrk_{\RR}(C_n)=(n+1)/2$. By 1) of Proposition~\ref{prop:properties} we have $\minrk_{\RR}^*(C_n)\le \alpha_f(C_n)=n/2<\minrk(C_n)$.
    If we let $w(x)=1/2$ for every vertex $x$ of $C_n$, then we can readily verify that $\{w(x):x\in V(C_n)\}$ is
    a feasible solution of~\eqref{PrimalEx}. It follows that $\minrk_{\RR}^*(C_n)\ge n/2$, and thus $\minrk_{\RR}^*(C_n)=n/2$.
%Indeed it is also simple to show that $\minrk_{\RR}^*(C_n)=n/2$ by definition.
}
\end{example}

Example~\ref{ex:G+C_n} provides a family of graphs where our bound outperforms simultaneously both $\minrk$ and Lov\'asz theta number, however it might seem a bit artificial since it is a disjoint union of two graphs. 
\begin{example}\label{ex:G+C_n}
  {\rm
    Let $G$ be the complement of the Schl\"afli graph. Then for odd $n\ge5$, by~Proposition~\ref{LovaszThetaNumberOfOddCycles}, Proposition~\ref{prop:LovaszProperty}, and Examples~\ref{ex:Schlaefli}-\ref{ex:C_n},
\begin{align*}
 \vartheta(G+C_n) &= \vartheta(G)+\vartheta(C_n)=9+\frac{n\cos(\pi/n)}{1+\cos(\pi/n)},\\
 \minrk(G+C_n) &= \minrk(G)+\minrk(C_n)=7+\frac{n+1}{2}.
\end{align*}
On the other hand, by 3) of Proposition~\ref{prop:properties},
\begin{align*}
  \minrk_{\RR}^*(G+C_n)=\minrk_{\RR}^*(G)+\minrk_{\RR}^*(C_n)\le 7+\frac{n}{2}.
\end{align*}
Hence $\Theta(G+C_n)\leq \minrk_{\RR}^*(G+C_n) \le 7+\frac{n}{2} < 7+\frac{n+1}{2} = \min\{\minrk(G+C_n),\vartheta(G+C_n)\}$.
}
\end{example}

\begin{figure}[t!]
  \centering
  \includegraphics[scale=0.5]{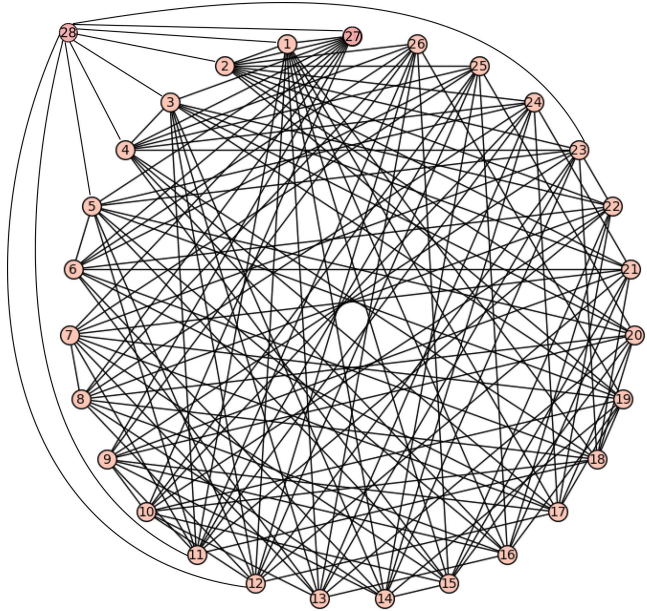}
  \caption{The graph $G$ in Example~\ref{SchlaefliModify}}
\end{figure}

In Example~\ref{SchlaefliModify} we construct a connected graph for which our bound also outperforms both $\minrk$ and Lov\'asz theta number. 
\begin{example}\label{SchlaefliModify}
  {\rm
    Let $G$ be the graph as plotted in Fig. 1. Note that $G_T$, the induced subgraph of $G$ on the vertices $T=\{1,2,\dots,27\}$, 
    is the complement of the Schl\"afli graph, and the vertex $28$ of $G$ is connected to vertices $1,\ldots,5,11,12,23$ and $27$. 
    Using Sagemath~\cite{sagemath} one can verify that $\vartheta(G)=9$. From Proposition~\ref{thm:MinrankOfModifiedSchlaefliGraph} in Appendix we see that $\minrk(G)=8$.
Take $f=\minrk_{\RR}$, and consider the following linear program:
\begin{align}\label{lp:ModifiedSchlaefliGraph}
  \begin{split}
 \textup{maximize } &\sum_{x} w(x) \\ 
 \textup{subject to }  &\sum_{x\in C}w(x)\le \minrk_{\RR}(C)= 1\textup{ for every clique $C$ in $G$},\\ 
 &\sum_{x\in T}w(x)\le \minrk_{\RR}(G_{T})= 7,\\ 
  & w(x)\ge 0.
 \end{split}
\end{align}
Using Sagemath~\cite{sagemath} one can compute that the optimal value of~\eqref{lp:ModifiedSchlaefliGraph} is equal to $71/9$.
Comparing~\eqref{lp:ModifiedSchlaefliGraph} with~\eqref{PrimalEx}, we have $\minrk_{\RR}^*(G)\le 71/9<8=\min\{\minrk(G),\vartheta(G)\}$.
}
\end{example}

The following result shows that we cannot always get a tighter bound through this linear programming variation.
\begin{proposition}\label{prop:minrkEx}
  Fix a field $\FF$. Let $G=(V,E)$ be a graph such that $\minrk_{\FF}(G)<\vartheta(G)$ and for any subset $S\subsetneq V$ we have $\minrk_{\FF}(G_S)\ge \vartheta(G_S)$. Then $\minrk_{\FF}^*(G)=\minrk_{\FF}(G)$.
\end{proposition}
\begin{proof}
  By definition $\minrk_{\FF}^*(G)\le\minrk_{\FF}(G)$. Now we show that $\minrk_{\FF}^*(G) \ge \minrk_{\FF}(G)$. Suppose that $\{q(S):S\subset V\}$ is an optimal solution for~\eqref{DualEx}. 
  It is easy to see that $q(V)\le 1$, otherwise $\{q(S):S\subset V\}$ is not an optimal solution for~\eqref{DualEx}. We have
\begin{align*}
  \minrk_{\FF}^*(G) &= \sum_{S}q(S)\minrk_{\FF}(G_S)\\
  &= q(V)\minrk_{\FF}(G)+\sum_{S\subsetneq V}q(S)\minrk_{\FF}(G_S)\\
  &\ge q(V)\minrk_{\FF}(G)+\sum_{S\subsetneq V}q(S)\vartheta(G_S).
\end{align*}
By Proposition~\ref{ThetaEx} below we have 
$$\vartheta(G)=\vartheta^*(G)\le q(V)\vartheta(G)+\sum_{S\subsetneq V}q(S)\vartheta(G_S).$$
Hence  
\begin{align*}
  \minrk_{\FF}^*(G) 
  \ge q(V)\minrk_{\FF}(G)+(1-q(V))\vartheta(G)
  \ge\minrk_{\FF}(G).
\end{align*}
This concludes our proof.
\end{proof}

The following example shows that there exist graphs satisfying the conditions of Proposition~\ref{prop:minrkEx}.
\begin{example}{\rm
  Fix a field $\FF$. Let $G$ be a graph such that $\minrk_{\FF}(G)<\vartheta(G)$.
  If $\minrk_{\FF}(G_S)\ge \vartheta(G_S)$ for any subset $S\subsetneq V(G)$, 
  then $\minrk_{\FF}^*(G)=\minrk_{\FF}(G)$ by Proposition~\ref{prop:minrkEx}.
  Otherwise, let $S$ be a subset of $V(G)$ with the smallest size among those subsets 
  such that $\minrk_{\FF}(G_S)<\vartheta(G_S)$. Obviously, the induced subgraph $G_S$ satisfies
  the conditions of Proposition~\ref{prop:minrkEx}, hence $\minrk_{\FF}^*(G_S)=\minrk_{\FF}(G_S)$.
  (Note that there are many graphs for which $\minrk_{\FF}(G)<\vartheta(G)$, e.g. the complement of the Schl\"afli graph for $\FF=\RR$.) 
  %From Example~\ref{ex:Schlaefli} we have $\minrk_{\RR}(G)<\vartheta(G)$. 
  %If $\minrk_{\RR}(G_S)\ge \vartheta(G_S)$ for any subset $S\subsetneq V(G)$, 
  %then $\minrk_{\RR}^*(G)=\minrk_{\RR}(G)$ by Proposition~\ref{prop:minrkEx}.
  %Otherwise, let $S$ be a subset of $V(G)$ with the smallest size among those subsets 
  %such that $\minrk_{\RR}(G_S)<\vartheta(G_S)$. Obviously, the induced subgraph $G_S$ satisfies
  %the conditions of Proposition~\ref{prop:minrkEx}, hence $\minrk_{\RR}^*(G_S)=\minrk_{\RR}(G_S)$.
}
\end{example}

\subsection{Bounds for Disjoint Union of Graphs}
For the Shannon capacity of the disjoint union of two graphs, we have the following simple observation.
\begin{corollary}
 $\Theta(G+H) \le \min\{\minrk(G)+\alpha_f(H),\alpha_f(G)+\minrk(H)\}.$
\end{corollary}
\begin{proof}
 Suppose $\minrk(G)=\minrk_{\FF}(G)$ for some field $\FF$. By Theorem~\ref{NewBound} and Proposition~\ref{prop:properties},
 we have $\Theta(G+H)\le\minrk_{\FF}^*(G+H)=\minrk_{\FF}^*(G)+\minrk_{\FF}^*(H)\le\minrk(G)+\alpha_f(H)$. Similarly, we can prove that $\Theta(G+H)\le\alpha_f(G)+\minrk(H).$ This concludes the proof.
\end{proof}

Next, we shall combine the Lov\'asz theta number and $\minrk_{\FF}$ through a weighted geometric mean
to get another upper bound on the Shannon capacity of the disjoint union.
Fix a field $\FF$ and suppose $0\le a\le 1$. Then we can easily verify that 
$$\vartheta^a\minrk_{\FF}^{*1-a}(G):=\vartheta(G)^a \minrk^*_{\FF}(G)^{1-a}$$
is also a submultiplicative upper bound on the independence number. 

\begin{corollary}\label{cor:GeometricAverage}
For a fixed field $\FF$ and a number $a\in [0,1]$, 
  $$\Theta(G+H)\le \vartheta^a\minrk_{\FF}^{*1-a}(G)+\vartheta^a\minrk_{\FF}^{*1-a}(H).$$
\end{corollary}
\begin{proof}
As $\vartheta^a\minrk_{\FF}^{*1-a}$ is a submultiplicative upper bound on the independence number, 
by Theorem~\ref{NewBound} we have
\begin{align*}
 \Theta(G+H) &\le (\vartheta^a\minrk_{\FF}^{*1-a})^*(G+H) \\
 &\le (\vartheta^a\minrk_{\FF}^{*1-a})^*(G)+(\vartheta^a\minrk_{\FF}^{*1-a})^*(H) \\
 &\le \vartheta^a\minrk_{\FF}^{*1-a}(G)+\vartheta^a\minrk_{\FF}^{*1-a}(H).
 \end{align*}
 Here the second inequality follows from (2) of Proposition~\ref{prop:properties}.
\end{proof}

\begin{example}\label{G+7C_5}
{\rm
Let $G$ be the complement of the Schl\"afli graph. Consider the graph $H=G+7 C_5$.
It is not hard to verify that $\vartheta(H)=9+7\sqrt{5}$ and $\minrk_{\RR}(H)=28$.
By Corollary~\ref{cor:GeometricAverage}, 
  \begin{align*}
  \Theta(H) 
  &\le \vartheta^a\minrk_{\RR}^{*1-a}(G)+7\cdot \vartheta^a\minrk_{\RR}^{*1-a}(C_5)\\
  &= \vartheta(G)^a \cdot \minrk^*_{\RR}(G)^{1-a}+7\cdot \vartheta(C_5)^{a}\cdot\minrk^*_{\RR}(C_5)^{1-a}\\
  &\le 9^{a}7^{1-a}+7(\sqrt{5})^{a}\left(\frac{5}{2}\right)^{1-a}.
  \end{align*}
For $a=0.287291$ the term $9^{a}7^{1-a}+7(\sqrt{5})^{a}\left(\frac{5}{2}\right)^{1-a}=24.4721$ achieving its minimum value on $[0,1]$. 
Note that this value is strictly better than $\vartheta(H)$ ($a=1$) and $\minrk_{\RR}(H)$.
}
\end{example}

Lastly, if we take $f$ to be the Lov\'asz theta number, our new bound cannot improve it.
\begin{proposition}\label{ThetaEx}
  $\vartheta^*(G)=\vartheta(G)$.
\end{proposition}
\begin{proof}
  From the primal linear program~\eqref{PrimalEx} we immediately get $\vartheta^*(G)\le\vartheta(G)$. On the other hand, let $\bfd$ and $\{\bfu_1,\dots,\bfu_m\}$ be an optimal solution for~\eqref{LovaszDual}. For each vertex $i,1\le i\le m$, we set $w(i)=(\bfd^T\bfu_i)^2$. Then for each subset $S$ of $V(G)$, we have 
  $$\sum_{i\in S}w(i)=\sum_{i\in S}(\bfd^T\bfu_i)^2\le\vartheta(G_S)$$
  by Proposition~\ref{LovaszThetaDuality}. Hence $\{w(i):1\le i\le m\}$ is a feasible solution for~\eqref{PrimalEx}, 
  and $\vartheta(G)=\sum_{1\le i\le m}w(i)\le \vartheta^*(G)$.
\end{proof}

\section{A New Upper Bound for Index Coding}\label{sec:IndexCoding}
In this section we show that our technique also allows us to derive a new bound for the Index Coding problem to be defined next. 
In the Index Coding problem, a sender holds a set of messages to be broadcast to a group of receivers.
Each receiver is interested in one of the messages, and has some prior side information
comprising some subset of the other messages. This variant of source coding problem was first proposed in~\cite{BirkKol2006} by Birk and Kol,
and later investigated in~\cite{Bar-Yossefetal2011} by Bar-Yossef et al.

The Index Coding problem can be formalized as follows: the sender holds $m$ messages $x_1,x_2,\ldots,x_m\in\Sigma$ 
where $\Sigma$ is the set of possible messages, and wishes to send them to $m$ receivers $R_1,R_2,\ldots,R_m$. Receiver $R_j$ wants to receive the message $x_{j}$, and knows some subset $N(j)$ of the other messages. The goal is to construct an efficient encoding scheme $\cE:\Sigma^m\rightarrow\Omega$, where $\Omega$ is a finite alphabet to be transmitted by the sender, such that for any $(x_1,x_2,\ldots,x_m)\in\Sigma^m$, every receiver $R_j$ is able to decode the message $x_{j}$ from the value $\cE(x_1,x_2,\ldots,x_m)$
together with his own side information $N(j)$. We associate a directed graph $G$ with the side-information subset $N(j)$, whose vertex set is $[m]=\{1,2,\ldots,m\}$, and whose edge set consists of all ordered pairs $(i,j)$ such that $x_j\in N(i)$. Here and in what follows, we further assume that the side-information graph $G$ is undirected, that is, if $x_j\in N(i)$ then $x_i\in N(j)$. For messages that are $t$ bits long, i.e. $|\Sigma|=2^t$, we use $\beta_t(G)$ to denote the corresponding minimum possible encoding length $\lceil\log_2{|\Omega|}\rceil$. The \emph{broadcast rate} of the side-information graph $G$ is defined as 
\begin{align*}
  \beta(G)  
  :=\inf_{t}\frac{\beta_t(G)}{t}
  =\lim_{t\rightarrow\infty}\frac{\beta_t(G)}{t},
\end{align*}
where the limit exists by subadditivity of $\beta_t(G)$ and Fekete's Lemma. That is to say that 
$\beta(G)$ is the average asymptotic number of broadcast bits needed per bit of input. This quantity has received significant interest, and in this section we prove a new upper bound for it. In~\cite{BirkKol2006,Bar-Yossefetal2011,LubetzkyStav2009}, it was proved that
\begin{align}\label{IndexCodingUpperBounds}
  \alpha(G)\le \beta(G)\le \minrk_{\FF}(G)\le \chrnb(G)
\end{align}
(here $\FF$ is an arbitrary finite field and $\chrnb(G)$ is the clique-cover number of $G$). On the other hand, Blasiak et al.~\cite{BlasiakKleinbergLubetzky2013} proved that $\beta(G)\le \alpha_f(G)$.
For more background and details on the Index Coding problem,
see~\cite{BirkKol2006,Bar-Yossefetal2011,Alonetal2008,BlasiakKleinbergLubetzky2013} and references therein.

Similarly as in Section~\ref{sec:new bound}, let $f$ be a real-valued function defined on graphs, and let $f^*(G)$ be the optimal value of~\eqref{PrimalEx}. 
Now we show that if $f$ is an upper bound on the broadcast rate, that is, $\beta(G)\le f(G)$ for any graph $G$, then $f^*$ is also an upper bound on the broadcast rate. The proof is a simple extension of~\cite[Claim 2.8]{BlasiakKleinbergLubetzky2013}.

\begin{theorem}\label{IndexCodingNewBound}
  If $f$ is an upper bound on the broadcast rate, then so is $f^*$.
\end{theorem}
\begin{proof}
Let $q^*$ be an optimal solution of the linear program~\eqref{DualEx}.
Without loss of generality, we can assume that each $q^*(S), S\subset V(G)$
is a nonnegative rational number, otherwise we can choose a rational number arbitrarily close to $q^*(S)$. 
By~\eqref{DualEx} we get $f^*(G)=\sum_{S}q^*(S)f(G_S)$ and $\sum_{S\owns x}q^*(S) \ge 1$ for every vertex $x$ in $G$. Let $t$ be a positive integer such that all the numbers $t\cdot q^*(S)$ are integers, and let $y_{S}=t\cdot q^*(S)$ for each $S\subset V(G)$. Then
  \begin{align*}
    t\cdot f^*(G)=\sum_{S}y_{S}f(G_S) \text{ and } \sum_{S\owns x}y_S \ge t \text{ for every vertex } x \text{ in } G.  
  \end{align*}
Namely, we cover the graph $G$ using a collection of  $y_S$ copies of $S$ for each $S\subset V(G)$. 
%Also, we can assume that $\sum_{S\owns x}y_S=t$ for each vertex $x$ in $G$, otherwise we replace each $S$ by a proper subset of $S$. 
Set $p=\sum_{S}y_{S}$. Then, altogether we have a sequence of $p$ subsets $S_1,S_2,\dots,S_p$, in which each $S\subset V(G)$ appears $y_S$ times, such that every vertex in $G$ appears in at least $t$ of these subsets. By assumption, for each induced side-information graph $G_{S_i} (1\le i\le p)$, the average asymptotic number of broadcast bits needed per bit of input is upper bounded by $f(G_{S_i})$.
Concatenating these $p$ individual index codes for the graph $G_{S_i}$ (if for some vertex $x,\sum_{S\owns x}y_S>t$, then we may ignore extra bits), we can see that the average asymptotic number of broadcast bits needed per bit of input for graph $G$ is upper bounded by $\sum_{i}f(G_{S_i})/t=\sum_{S}y_{S}f(G_S)/t =f^*(G)$.
This concludes the proof.
\end{proof}

Let us now consider the function $\displaystyle f(G)=\inf_{\FF} \minrk_{\FF}(G)$ for any graph $G$, where the infimum ranges over all finite fields $\FF$. 
By~\eqref{IndexCodingUpperBounds} we see that $f$ is an upper bound on the broadcast rate, hence so is $f^*$ by Theorem~\ref{IndexCodingNewBound}. 
Note that for different graphs, the value of $f$ may be obtained as the minimum rank over different fields. 
Therefore, the achievable scheme given by an optimal solution of the corresponding linear program~\eqref{DualEx} might yield a scheme that uses several different fields simultaneously.
More simply, we can take $f(G)=\minrk_{\FF}(G)$ for some fixed finite field $\FF$. As $\minrk_{\FF}$ is an upper bound for the broadcast rate by~\eqref{IndexCodingUpperBounds}, we can get the following result directly from Theorem~\ref{IndexCodingNewBound}.
\begin{corollary}
  For any graph $G$ and any finite field $\FF$, $\beta(G)\le\minrk_{\FF}^*(G)$.    
\end{corollary}
By 1) of Proposition~\ref{prop:properties}, $\minrk_{\FF}^*(G)\le\alpha_f(G)$. Hence the bound $\minrk_{\FF}^*$ is at least as good as $\minrk_{\FF}$ and $\alpha_f$. The following example shows that sometimes $\minrk_{\FF}^*$ can simultaneously outperform both $\minrk_{\FF}$ and $\alpha_f$.

\begin{example}\label{IndexCodingSchlaefliModify}
  {\rm
    Let $G$ and $T$ be defined as in Example~\ref{SchlaefliModify}. From Example~\ref{ex:Schlaefli} we have $\minrk_{\FF_{11}}(G_{T})=7$. Similarly as in Example~\ref{SchlaefliModify} we get $\minrk_{\FF_{11}}^*(G)\le 71/9$, which is better than $\minrk_{\FF}$ (as $\minrk(G)=8$) and $\alpha_f$ (as $\alpha_f(G)\ge\vartheta(G)=9$).
}
\end{example}

\begin{remark}\label{FractionalHaemersBound}
{\rm
    Blasiak~\cite{Blasiak2013} and Shanmugam et al.~\cite{ShanmugamAsterisDimakis2015} independently\footnote{Here we adopt the notion in Blasiak~\cite{Blasiak2013}, which is slightly different from that in Shanmugam et al.~\cite{ShanmugamAsterisDimakis2015}.} obtained expressions for the infimum of the broadcast rate of vector linear broadcasting schemes~\footnote{A vector linear broadcasting scheme over a finite field $\FF$ is a scheme in which the message alphabet $\Sigma$ is a finite dimensional vector space over $\FF$ and the encoding and decoding functions are linear.} over all finite fields as follows. 
    Let $G$ be a graph with vertex set $V(G)=\{1,2,\dots,m\}$, and let $B$ be an $m\times m$ matrix whose entries are $k\times k$ matrices over some field $\FF$. We say that $B$ \emph{fractionally represents} the side-information graph $G$ over $\FF^k$ if $B_{ii}$ is the identity matrix of size $k$, and $B_{ij}$ is the zero matrix of size $k$ whenever $i$ and $j$ are nonadjacent.
    The \emph{fractional minrank} of $G$ is defined by
    \begin{align}\label{FractionalMinrank}
    \minrk_{f}^{\FF}(G):= \inf_{k}\frac{\min\{\rank(B): B \text{ fractionally represents } G \text{ over } \FF^k\}}{k}
    \end{align}
    and
    $$\minrk_{f}(G):=\inf_{\FF}\minrk_f^{\FF}(G).$$
    It is shown in~\cite{Blasiak2013,ShanmugamAsterisDimakis2015} that $\minrk_{f}^{\FF}(G)$ is the infimum of the broadcast rate of all vector linear broadcasting schemes over $\FF$. 
    On the other hand, we can obtain a vector linear broadcasting scheme over $\FF$ of rate  $\minrk_{\FF}^*(G)$, by using a vector linear broadcasting scheme of rate $\minrk_{\FF}(G_{S_i})$ for each induced subgraph $G_{S_i}$ in the proof of Theorem~\ref{IndexCodingNewBound}.
    Hence $\beta(G)\le\minrk_{f}^{\FF}(G)\le\minrk_{\FF}^*(G)$. 
    Note that it is very difficult to compute $\minrk_{f}^{\FF}(G)$ via~\eqref{FractionalMinrank}.
    But our graph invariant $\minrk_{\FF}^*(G)$ provides a way to approach $\minrk_{f}^{\FF}(G)$,
    since we can always get an upper bound for $\minrk_{\FF}^*(G)$, and thus for $\minrk_{f}^{\FF}(G)$,
    by solving the linear programming problem~\eqref{DualEx} or its subproblems obtained by removing some constraints from~\eqref{DualEx}.
    Blasiak~\cite{Blasiak2013} and Shanmugam et al.~\cite{ShanmugamAsterisDimakis2015} also proved that $\Theta(G)\le\minrk_{f}^{\FF}(G)$.
    See~\cite{BukhCox2018} for more properties of $\minrk_{f}^{\FF}$.
    }
\end{remark}

\begin{remark}\label{BukhCox}
{\rm
In~\cite{BukhCox2018} Bukh and Cox asked the following question:
%\begin{itemize}
Are there graphs for which $\vartheta(G) < \minrk(G)$, yet $\minrk_f(G) < \vartheta(G)$?
%\end{itemize}
Example~\ref{G+7C_5} gives an affirmative answer to this question.
Recall that we let $G$ be the complement of the Schl\"afli graph and consider the graph $G+7 C_5$.
Then $\vartheta(G+7C_5)=9+7\sqrt{5}<\minrk(G+7C_5)=7+7\cdot 3=28$.
On the other hand, $\minrk_f(G+7C_5)\leq \minrk_{\FF_{11}}^*(G+7C_5) \leq 7 + 7\cdot 2.5 = 24.5$ which is strictly less than $\vartheta(G+7C_5)$.
}
\end{remark}

\begin{remark}{\rm
    Shanmugam, Dimakis and Langberg~\cite{Shanmugametal2014} presented an upper bound for the broadcast rate of general side-information graphs using the local chromatic number. Later, this bound is further extended by Arbabjolfaei and Kim~\cite[Theorems 3--4]{ArbabjolfaeiKim2014} and Agarwal and Mazumdar~\cite[Theorems 3--5]{AgarwalMazumdar2016} via linear programming. Similarly as in Remark~\ref{remark:1}, for undirected graphs, it is not hard to check that those bounds are always no stronger than the fractional independence number.
}
\end{remark}

\section*{Appendix}
\begin{lemma}\cite[Theorem 3.6]{Tims2013}\label{MinrankTechnique}
  Let $G$ be a graph, let $I=\{v_1,\dots,v_k\}$ be a maximum independent set in $G$, and let $u$ be a vertex not in $I$. Set $J=N(u)\cap I$, which is nonempty since $I$ is maximum. If there exists another vertex $w\in V(G)\backslash (I\cup\{u\})$ that is adjacent to $u$ but not adjacent to any vertex of $J$, then delete the edge $(u,w)$, and let $H$ be the resulting spanning subgraph of $G$. Then 
  $$\minrk(G)=\alpha(G) \text{ if and only if } \minrk(H)=\alpha(G).$$
\end{lemma}
\begin{proof}
  (This proof was given by Tims in~\cite{Tims2013}).
  It is easy to see that $\alpha(G)\le\minrk(G)\le\minrk(H)$. Hence if $\minrk(H)=\alpha(G)$, then $\minrk(G)=\alpha(G)$.
  Now we assume that $\minrk(G)=\alpha(G)=k$, and let $B$ be a matrix that fits $G$ with $\rank(B)=k$. 
  Without loss of generality, we can assume that $J=\{v_1,v_2,\ldots,v_l\}$ where $1\le l\le k$,
  and all digonal entries of the matrix $B$ are equal to $1$.
  Then we can write the matrix $B$ as follows:
  $$
  \kbordermatrix{
    \mbox{}&v_1&v_2&\ldots&v_l&v_{l+1}&\ldots&v_k&u&w\\
    v_1    &1  &0  &\ldots&0  &0      &\ldots&0  &*&0\\
    v_2    &0  &1  &\ldots&0  &0      &\ldots&0  &*&0\\
    \vdots &\vdots&\vdots&\ddots&\vdots&\vdots&\ddots&\vdots&\vdots\\
    v_l    &0  &0  &\ldots&1  &0      &\ldots&0  &*&0\\
    v_{l+1}&0  &0  &\ldots&0  &1      &\ldots&0  &0&*\\
    \vdots &\vdots&\vdots&\ddots&\vdots&\vdots&\ddots&\vdots&\vdots\\
    v_{k}  &0  &0  &\ldots&0  &0      &\ldots&1  &0&*\\
    u      &*  &*  &\ldots&*  &0      &\ldots&0  &1&*\\
    w      &0  &0  &\ldots&0  &*      &\ldots&*  &*&1\\
  }.
  $$
  (Here the entry indicated by $*$ can be any value, and we only present part of the matrix $B$.)
  Since $\rank(B)=k$ and the first $k$ rows of $B$ are independent, all the rows of $B$
  can be written as linear combinations of the first $k$ rows. In particular, we can easily verify that
  the row indicated by $u$ must be a linear combination of the first $l$ rows, and hence $B_{uw}=0$.
  Similarly, we have $B_{wu}=0$. Therefore matrix $B$ also fits $H$, and it follows that $\minrk(H)=k=\alpha(G)$.
\end{proof}

\begin{proposition}\label{thm:MinrankOfModifiedSchlaefliGraph}
  Let $G$ be the graph as plotted in Figure 1. Note that the induced subgraph of $G$ on the vertices $T=\{1,2,\dots,27\}$ is the complement of the Schl\"afli graph. Then $\minrk(G)=8$.
\end{proposition}
\begin{proof}
  It can be checked that the set $I=\{8, 9, 13, 15, 19, 25, 28\}$ is a maximum independent set of $G$. Hence $\minrk(G)\ge|I|=7$. On the other hand, from the minimum rank of the complement of the Schl\"afli graph (see~Example~\ref{ex:Schlaefli}), we have $\minrk(G)\le 7+1 =8$. 
  Now we use Lemma~\ref{MinrankTechnique} to show that $\minrk(G)=8$.
  First, let $u=6$. Then the neighbors of $u$ are $N(6)=\{5, 13, 14, 17, 18, 21, 22, 25, 26,27\}$, and $J=N(6)\cap I=\{13, 25\}$. Set $W =\{5, 17, 18, 21, 22,27\}$. We can check that every vertex $w$ in $W$ satisfies the conditions in Lemma~\ref{MinrankTechnique}. Secondly, let $u=17$. Then $N(17)= \{1, 4, 7, 9, 12, 19, 20, 22, 24\}$, and $J=N(17)\cap I=\{9, 19\}$. Set $W=\{1, 4, 12, 22, 24\}$. We can also check that every vertex $w$ in $W$ satisfies the conditions in Lemma~\ref{MinrankTechnique}.  
  Lastly, we delete the edges $(6,5), (6,17), (6,18), (6,21), (6,22), (6,27),  (17,1), (17,4), (17,12), (17,22), (17,24)$ from $G$, and let $H$ be the resulting spanning subgraph of $G$. It can be checked that the set $\{6, 12, 15, 16, 17, 18, 24,27\}$ is a maximum independent set of $H$. Therefore, $\minrk(H)\ge8>\alpha(G)$, and hence $\minrk(G)>\alpha(G)=7$ by Lemma~\ref{MinrankTechnique}. This concludes the proof.
\end{proof}

\bibliographystyle{siamplain}
\bibliography{ZeroError}
\end{document}